\documentclass[conference]{IEEEtran}
\IEEEoverridecommandlockouts
\usepackage{cite}
\usepackage{amsmath,amssymb,amsfonts}
\usepackage{amsthm}
\usepackage{algorithmic}
\usepackage{graphicx}
\usepackage{textcomp}
\usepackage{xcolor}
\def\BibTeX{{\rm B\kern-.05em{\sc i\kern-.025em b}\kern-.08em
    T\kern-.1667em\lower.7ex\hbox{E}\kern-.125emX}}
\newtheorem{theorem}{Theorem}
\newtheorem{lemma}{Lemma}

\newtheorem{proposition}{Proposition}
\newtheorem{definition}{Definition}
\newtheorem{remark}{Remark}

\usepackage{booktabs}
\usepackage{multirow}
\usepackage[ruled,vlined,linesnumbered]{algorithm2e}
\SetKwInput{KwInput}{Input}
\SetKwInput{KwOffline}{Offline}
\SetKwInput{KwOnline}{Online}
\SetKwInput{KwOutput}{Output}
\DeclareMathOperator{\col}{col}
\begin{document}

\title{Fast-DeePC: Fast Data-Enabled Predictive Control using Soft-constrained Scheme}

\author{Terry Zhang$^{1,*}$,
Jan Drgona$^1$,
Malika Sader$^2$, Di Liu$^3$,
Liang Wu$^{1,*,\dagger}$
\thanks{This research is partially supported by the U.S. DOE, Office of Science, ASCR program under the Scientific Discovery through Advanced Computing (SciDAC) Institute “LEADS: LEarning-Accelerated Domain Science” and by the Ralph O’Connor Sustainable Energy Institute (ROSEI) at Johns Hopkins University.

$^{*}$Equal Contribution. $^{\dagger}$Corresponding Author: Liang Wu (\texttt{wliang14@jh.edu}). $^{1}$Johns Hopkins University, MD 21218, USA. $^{2}$Beijing Information Science and Technology University, Beijing 102206, CHN. $^{3}$Southeast University, Nanjing 211189, CHN.}
}

\maketitle

\begin{abstract}
Data-enabled predictive control (DeePC) replaces the explicit system model of model predictive control with a data-driven predictor built directly from recorded input-output trajectories, but the original hard-equality formulation keeps its coordinate vector as an online decision variable whose size grows with the offline data length, so the online problem grows more expensive as more data is used to enrich the behavioral representation. This article proposes Fast-DeePC, a framework which relaxes this consistency relation into the control objective with a quadratic ridge on the coordinate vector in place of the customary sparsity-promoting penalty, admitting a closed-form, offline elimination of the coordinate vector and reducing the online problem to a fixed-size box-constrained quadratic program solved by a structure-exploiting interior-point method. On identical soft quadratically-regularized problems, the result agrees with a general-purpose solver’s solution and is up to two orders of magnitude faster.
\end{abstract}

\begin{IEEEkeywords}
Data-enabled predictive control, quadratic programming, interior-point method
\end{IEEEkeywords}

\section{Introduction}
\label{sec:introduction}

\IEEEPARstart{C}{lassical} control practice like model predictive control (MPC), has traditionally been built on the explicit mathematical description of the plant~\cite{rawlings2020model}. Its effectiveness depends on whether a compact, sufficiently accurate model can be learned from data. However, for many modern engineering systems, building such a model is often the most challenging part. Strong nonlinearities, large state dimensions, and the high cost of dedicated system identification experiments can make obtaining an accurate parametric model considerably more difficult than exploiting the abundant operational data that are already routinely collected \cite{hjalmarsson2005experiment,ogunnaike1995contemporary}. This has motivated an alternative line of control design to utilize measured trajectories directly to compute a control action. 

Recently, data-enabled predictive control (DeePC) \cite{coulson2019data} has attracted increasing attention as a model-free predictive control framework. DeePC computes optimal control actions directly from measured input--output data without explicitly identifying a parametric system model. Built upon Willems' Fundamental Lemma from behavioral systems theory \cite{willems1986time,willems2005note,markovsky2006exact}, it represents future trajectories as linear combinations of previously collected data through the column space of a Hankel matrix constructed from a persistently exciting input--output trajectory. By bypassing the model identification stage while retaining the predictive control formulation, DeePC is particularly attractive when accurate system models are difficult to obtain. Under noise-free measurements and a persistently exciting experiment, this behavioral formulation is provably equivalent to model-based MPC with an exact state estimate \cite{coulson2019data,markovsky2008data}, thereby eliminating the need for explicit state-space realization while retaining the optimization structure of MPC.

\subsection{Related Work}
A large body of research has focused on improving the robustness and theoretical properties of DeePC while retaining its behavioral formulation. The original regularized DeePC relaxes the behavioral consistency constraints through an $\ell_1$  regularization on the behavioral coefficients together with an output slack variable, later interpreted from a distributionally robust optimization perspective \cite{coulson2019regularized,coulson2021distributionally}. Ref. \cite{huang2021quadratic} proposed a \emph{quadratic} (namely $\ell_2$) regularization motivated by measurement disturbances, and further extended this framework to robust min--max formulations and kernel-based nonlinear DeePC \cite{huang2023robust}. Ref.  \cite{dorfler2022bridging} interpreted quadratic regularization as an implicit system identification step within a convex relaxation of indirect identify-then-control approaches. Other works retain the exact behavioral constraints while incorporating terminal ingredients and Lyapunov arguments to establish recursive feasibility, stability, and robustness guarantees \cite{berberich2020data,berberich2020robust}.

Despite its conceptual elegance and strong theoretical guarantees, the computational complexity of DeePC grows with the amount of stored data, since future trajectories are parameterized directly by previously collected data rather than by a fixed-order system model. As more data are incorporated to improve the behavioral representation, the dimension of the resulting online optimization problem increases accordingly. In contrast, model-based MPC compresses the available data into a compact model offline, making its online computational complexity largely independent of the size of the identification dataset. Moreover, DeePC requires the behavioral consistency constraints to be enforced during every online optimization, further increasing the computational burden.

Only a few recent works explicitly target this computational bottleneck. Deep DeePC \cite{zhang2025deep} replaces the online quadratic program (QP) with a neural network policy learned offline, significantly reducing online computational cost at the expense of the convex optimization structure and formal guarantees on feasibility and optimality. RDS-DeePC \cite{li2025rds} instead reduces the effective Hankel dimension by selecting representative trajectory segments using an influence-function-based sensitivity criterion. Ref. \cite{shi2024efficient} proposed an efficient low-dimensional DeePC variant via the use of singular value decomposition. Similarly, Ref. \cite{faye2024computationally} also proposed a computationally efficient reformulation of DeePC via range space equivalence, which only involves the replacement of two matrices, but relies on problems of much smaller dimension.

Predictor-based data-driven control methods are used for fixed-dimensional online optimization. SPC~\cite{favoreel1999spc} projects future outputs onto the row span of past input--output data and future inputs, leaving the future input as the online decision variable, and $\gamma$-DDPC~\cite{breschi2023data} generalizes this via an LQ decomposition that separates initial-condition fitting from future-performance optimization. Both rely on first constructing a projected predictor or an LQ-based coordinate decomposition, and \cite{breschi2023data} further shows that regularizing the full behavioral coordinate requires careful tuning.But none of these reductions eliminates the behavioral coordinate from an \emph{exact} soft-constrained DeePC objective while guaranteeing a unique solution independent of data rank.

% Unlike these approximation- or data-reduction-based approaches, the proposed method fundamentally removes the data-dependent behavioral parameterization from the soft-constrained DeePC formulation through an exact analytical elimination. Consequently, the size of the online optimization problem becomes independent of the amount of stored data, without resorting to approximation, learning, or trajectory selection.

% In contrast to both, this article targets the coupling not by approximating $g$ with a learned surrogate or by shrinking $T$ offline, but by eliminating $g$ \emph{exactly} from a soft-constrained DeePC objective, so that the online problem size becomes scalable by construction rather than by approximation or learning. 

% Despite these advances, the underlying behavioral parameterization remains unchanged. Consequently, the online optimization problem continues to scale with the amount of stored data, since the behavioral representation is still parameterized directly by the Hankel matrix. 

\subsection{Contribution}
 
% This article proposes \emph{Fast-DeePC}, a soft-constrained scheme of DeePC with a dedicated solver. The proposed method makes the following contributions:(i) The behavioral consistency relation is relaxed into the objective through quadratic penalties and a ridge, allowing the coordinate vector to be eliminated exactly by partial minimization and reducing the online problem to a fixed-dimensional QP. We ensure that (ii)~The reduced Hessian is positive definite for any recorded data, so the online problem always admits a unique minimizer. (iii)~The reduced problem is solved by a feasible Mehrotra predictor--corrector interior-point method implemented in C, which, on three different  benchmarks, matches a general-purpose solver on the identical soft problem while reducing online solve time by up to two orders of magnitude.

This article presents \emph{Fast-DeePC}, a soft-constrained DeePC formulation equipped with a dedicated solver. The main contributions are threefold:
(i)~We relax the behavioral consistency relation through quadratic penalties and introduce ridge regularization, enabling exact elimination of the behavioral coordinate vector by partial minimization. This reduces the online optimization problem to a fixed-dimensional quadratic program (QP).
(ii)~We establish that the reduced Hessian is positive definite for arbitrary recorded data, guaranteeing a unique minimizer whenever the online problem is feasible.
(iii)~We develop a C implementation of a feasible Mehrotra predictor--corrector interior-point method tailored to the reduced QP. Across three benchmarks, the proposed solver achieves solution quality comparable to that of a general-purpose solver on the same soft-constrained problem, while reducing online solve time by up to two orders of magnitude.

\section{Preliminaries: DeePC}
\label{sec:preliminaries}

\subsection{Behavioral System Representation}

Consider a discrete-time LTI system
\begin{equation}
x_{t+1}=Ax_t+Bu_t,\qquad y_t=Cx_t+Du_t,
\label{eq:lti}
\end{equation}
with $A\in\mathbb{R}^{n\times n}$, $B\in\mathbb{R}^{n\times m}$, $C\in\mathbb{R}^{p\times n}$, $D\in\mathbb{R}^{p\times m}$, state $x_t\in\mathbb{R}^n$, input $u_t\in\mathbb{R}^m$, and output $y_t\in\mathbb{R}^p$. For an integer $k$, stacking the output map $C$ together with its successive shifts along the state trajectory, $\mathcal{O}_k(A,C):=\operatorname{col}(C,CA,\dots,CA^{k-1})$, gives a matrix whose rank determines how many consecutive output samples suffice to recover the state uniquely. The \emph{lag} $\nu$ of \eqref{eq:lti} is the smallest such $k$ for which $\mathcal{O}_k(A,C)$ attains full column rank $n$.

\begin{definition}[Persistency of excitation]
\label{def:pe}
For $w=\operatorname{col}(w_0,\dots,w_{T-1})$ with $w_i\in\mathbb{R}^q$,
the Hankel matrix of depth $L$ is
\begin{equation}
\mathcal{H}_L(w):=
\begin{bmatrix}
w_0 & w_1 & \cdots & w_{T-L}\\
\vdots & \vdots & & \vdots\\
w_{L-1} & w_L & \cdots & w_{T-1}
\end{bmatrix}
\in\mathbb{R}^{qL\times(T-L+1)},
\label{eq:hankel-def}
\end{equation}
and $w$ is \emph{persistently exciting} of order $L$ if
$\mathcal{H}_L(w)$ has full row rank.
\end{definition}

\begin{lemma}[Fundamental lemma\cite{willems2005note}]
\label{lem:fundamental}
Let \eqref{eq:lti} be controllable and let $u^{d}=\operatorname{col}(u^{d}_0,\dots,u^{d}_{T-1})\in\mathbb{R}^{mT}$ be persistently exciting of order $L+n$, with $y^{d}\in\mathbb{R}^{pT}$ the corresponding output sequence. Then $\operatorname{col}(u^s,y^s)$ is a length-$L$ trajectory of \eqref{eq:lti} if and only if there exists $g\in\mathbb{R}^{T-L+1}$ with
\begin{equation}
\begin{bmatrix}\mathcal{H}_L(u^{d})\\ \mathcal{H}_L(y^{d})\end{bmatrix}g
=\begin{bmatrix}u^s\\ y^s\end{bmatrix}.
\label{eq:fundamental}
\end{equation}
\end{lemma}

Lemma~\ref{lem:fundamental} identifies the subspace of all valid length-$L$ trajectories with the range of the Hankel matrix: a single sufficiently rich trajectory is itself a non-parametric predictive model.

\subsection{Data Collection and the Consistency Equation}

Fix a prediction horizon $N\in\mathbb{Z}_{>0}$ and an initialization horizon $T_{\mathrm{ini}}\in\mathbb{Z}_{>0}$, and let $T$ satisfy $T\ge(m+1)(T_{\mathrm{ini}}+N+n)-1$, a necessary condition for persistency of excitation of order $T_{\mathrm{ini}}+N+n$. Recording $u^{d},y^{d}$ once offline and partitioning the Hankel matrices into
\emph{past} and \emph{future} block rows gives
\begin{equation}
\begin{bmatrix} U_p\\ U_f \end{bmatrix}
:= \mathcal{H}_{T_{\mathrm{ini}}+N}(u^{d}),\qquad
\begin{bmatrix} Y_p\\ Y_f \end{bmatrix}
:= \mathcal{H}_{T_{\mathrm{ini}}+N}(y^{d}),
\label{eq:hankel-partition}
\end{equation}
with $U_p\in\mathbb{R}^{mT_{\mathrm{ini}}\times n_g}$, $U_f\in\mathbb{R}^{mN\times n_g}$ (similarly $Y_p,Y_f$) and $n_g:=T-T_{\mathrm{ini}}-N+1$ Hankel columns. Each column collects the past signal of length $T_{\mathrm{ini}}$ and the future signal of length $N$ within one recorded trajectory.

At time $t$, let $u_{\mathrm{ini}}=\operatorname{col}(u_{t-T_{\mathrm{ini}}},\dots,u_{t-1}) \in\mathbb{R}^{mT_{\mathrm{ini}}}$ denote the most recent measured inputs, $u=\operatorname{col}(u_t,\dots,u_{t+N-1})\in\mathbb{R}^{Nm}$ the future input sequence, and $y_{\mathrm{ini}}\in\mathbb{R}^{pT_{\mathrm{ini}}}$, $y\in\mathbb{R}^{Np}$ their output counterparts; $r=\operatorname{col}(r_t,\dots,r_{t+N-1})\in\mathbb{R}^{Np}$ is the stacked reference. By Lemma~\ref{lem:fundamental} with $L=T_{\mathrm{ini}}+N$, the signal $\operatorname{col}(u_{\mathrm{ini}},u,y_{\mathrm{ini}},y)$ is admissible if and only if there exists a \emph{coordinate vector} $g\in\mathbb{R}^{n_g}$ with
\begin{equation}
\begin{bmatrix} U_p\\ Y_p\\ U_f\\ Y_f \end{bmatrix} g
=
\begin{bmatrix}
u_{\mathrm{ini}}\\ y_{\mathrm{ini}}\\ u\\ y
\end{bmatrix}.
\label{eq:consistency}
\end{equation}
Provided $T_{\mathrm{ini}}\ge\nu$, the measured window $\xi_t:=\operatorname{col}(u_{\mathrm{ini}},y_{\mathrm{ini}})$ fixes a unique underlying state and $y$ is uniquely determined for every candidate future input $u$. The past block rows enforce agreement with the recent input--output window, the future-input block row fixes the candidate input sequence, and the future output is then read as $y=Y_fg$, without explicitly constructing a state estimate~\cite{coulson2019data,markovsky2008data}.

\subsection{The DeePC Algorithm}

Given constraint sets $\mathcal{U}\subseteq\mathbb{R}^m$, $\mathcal{Y}\subseteq\mathbb{R}^p$ and performance matrices $Q\in\mathbb{R}^{p\times p}$, $R\in\mathbb{R}^{m\times m}$, DeePC solves
\begin{equation}
\begin{aligned}
\underset{g,u,y}{\text{minimize}}\ \ &
\sum_{k=0}^{N-1}\bigl(\|y_k-r_{t+k}\|_Q^2+\|u_k\|_R^2\bigr)\\
\text{subject to}\ \ &
\eqref{eq:consistency},\quad
u_k\in\mathcal{U},\ y_k\in\mathcal{Y},\ k=0,\dots,N-1,
\end{aligned}
\label{eq:deepc-qp}
\end{equation}
in a receding-horizon manner: solve for $g^\star$, apply the first $s$ input samples $u_0^\star,\dots,u_{s-1}^\star$ of $u^\star=U_fg^\star$, update $\xi_t$, and repeat with $t\leftarrow t+s$; as in MPC, the shift $s\in\{1,\dots,N\}$ is a design parameter, taken as $s=1$ here. Note that $u$ and $y$ are not free decision variables, since \eqref{eq:consistency} pins them to $U_fg$ and $Y_fg$; the model and state estimate of MPC are replaced entirely by data samples~\cite{coulson2019data}.

Problem \eqref{eq:deepc-qp} therefore achieves identification, estimation, and prediction through a single linear equation, at the price of a QP whose decision variable $g$ has dimension $n_g$, growing with the offline data length $T$. Enriching the behavioral representation with additional data thus directly inflates the online problem. The scheme developed next removes this coupling by relaxing \eqref{eq:consistency} into the objective and eliminating $g$ offline.

\section{Method}
\label{sec:method}

We develop a Soft-constrained Scheme that moves the $T$-dependent computation associated with $g$ to the offline phase. The resulting online problem is a dense box-constrained QP in $z=\operatorname{col}(u,y)\in\mathbb{R}^{N(m+p)}$.

Let $N\in\mathbb{Z}_{>0}$ be the prediction horizon and $T_{\mathrm{ini}}\in\mathbb{Z}_{>0}$ the initialization horizon. The stacked
future trajectories $u$, $y$, and $r$, and the Hankel blocks $U_p,Y_p,U_f,Y_f$, are defined as in Section~\ref{sec:preliminaries}. Stacking the per-step cost matrices $Q$ and $R$ from
\eqref{eq:deepc-qp} gives
\begin{equation}
\begin{alignedat}{2}
W_y &:= \operatorname{blkdiag}(Q,\ldots,Q)\succeq0,
&\qquad W_y &\in \mathbb{R}^{Np\times Np},\\
W_u &:= \operatorname{blkdiag}(R,\ldots,R)\succeq0,
&\qquad W_u &\in \mathbb{R}^{Nm\times Nm}.
\end{alignedat}
\end{equation}
writes the tracking and effort cost of \eqref{eq:deepc-qp} compactly as $\tfrac12(y-r)^\top W_y(y-r)+\tfrac12 u^\top W_u u$.

To group the four blocks of \eqref{eq:consistency} into the part fixed by the current measurement and the part carrying the decision $u$, define the past-and-input regressor and its target
\begin{equation}\label{eq:Z-wt}
Z:=\begin{bmatrix} U_p\\ Y_p\\ U_f\end{bmatrix},
~w_t(u):=\begin{bmatrix} u_{\mathrm{ini}}\\ y_{\mathrm{ini}}\\ u\end{bmatrix}
       =\underbrace{\begin{bmatrix} u_{\mathrm{ini}}\\ y_{\mathrm{ini}}\\ 0\end{bmatrix}}_{=:\,w_{\mathrm{ini},t}}
        +\begin{bmatrix} 0\\ 0\\ I\end{bmatrix} u ,
\end{equation}
so that $w_t(u)$ is affine in $u$ and $w_{\mathrm{ini},t}$ depends only on the measured window $(u_{\mathrm{ini}},y_{\mathrm{ini}})$. The future-output block $Y_f$ is kept separate as the quantity to be predicted. Note that the four behavioral relations of \eqref{eq:consistency} are not penalized independently: three of them ($U_pg=u_{\mathrm{ini}}$, $Y_pg=y_{\mathrm{ini}}$, $U_fg=u$) are grouped into the single $\lambda$-weighted consistency term through the stacked regressor $Z$; only the fourth, output-prediction relation ($Y_fg=y$) is penalized separately, with weight $\rho$. It is also convenient to isolate the purely-measured part of the regressor,
\begin{equation}\label{eq:Zp-wp}
Z_{\mathrm p}:=\begin{bmatrix} U_p\\ Y_p\end{bmatrix},
\quad
w_{\mathrm p,t}:=\begin{bmatrix} u_{\mathrm{ini}}\\ y_{\mathrm{ini}}\end{bmatrix},
\end{equation}
so that $w_{\mathrm{ini},t}=\col(w_{\mathrm p,t},0)$; this form is used
below to construct the online linear term.

\subsubsection{The penalized program and offline elimination of $g$}

We relax the hard consistency constraint \eqref{eq:consistency} into the objective, penalizing the past-and-input mismatch $Zg-w_t(u)$ with weight $\lambda$, the output mismatch $Y_fg-y$ with weight $\rho$, and adding the ridge $\lambda_g\|g\|_2^2$ in place of the $\ell_1$ term of the regularized DeePC \cite[Eq.~(8)]{coulson2019data}, yielding
\begin{equation}
\begin{aligned}
\min_{g,u,y}\ \
& \tfrac12(y-r)^\top W_y(y-r)+\tfrac12 u^\top W_u u
  +\tfrac{\rho}{2}\,\|Y_f g - y\|_2^2 \\
& \qquad\qquad
  +\tfrac{\lambda}{2}\,\|Z g - w_t(u)\|_2^2
  +\tfrac{\lambda_g}{2}\,\|g\|_2^2 \\
\text{s.t.}\ \
& u\in\mathcal{U}^N,\quad y\in\mathcal{Y}^N .
\end{aligned}
\label{eq:penalized}
\end{equation}
Here $\lambda>0$ weights data consistency, $\rho>0$ weights the output prediction, and $\lambda_g>0$ is a ridge on $g$: larger $\lambda,\rho$ enforce consistency with the recorded data more tightly, while smaller values favor tracking and constraint satisfaction.

\begin{lemma}[Offline elimination of $g$]
\label{lem:elim}
Let $\lambda,\rho,\lambda_g>0$ and define
\begin{equation}
M:=\lambda\,Z^\top Z+\rho\,Y_f^\top Y_f+\lambda_g I .
\label{eq:M}
\end{equation}
Then $M\succ0$, and for any fixed $(u,y)$ the inner problem
$\min_g\{\tfrac{\rho}{2}\|Y_fg-y\|_2^2+\tfrac{\lambda}{2}\|Zg-w_t(u)\|_2^2
+\tfrac{\lambda_g}{2}\|g\|_2^2\}$ has the unique minimizer
\begin{equation}
g^\star(u,y)=M^{-1}\big(\lambda\,Z^\top w_t(u)+\rho\,Y_f^\top y\big).
\label{eq:gstar}
\end{equation}
The matrix $M$ (equivalently its Cholesky factor) depends only on the offline
data and the fixed weights and is computed once, before control begins.
\end{lemma}

\begin{proof}
$Z^\top Z\succeq0$ and $Y_f^\top Y_f\succeq0$, so $M\succeq\lambda_g I\succ0$ for
$\lambda_g>0$. The inner objective is a strongly convex quadratic in $g$;
equating its gradient to zero gives the normal equations
$Mg=\lambda Z^\top w_t(u)+\rho Y_f^\top y$, whose unique solution is
\eqref{eq:gstar}. The matrix $M$ contains no online quantities.
\end{proof}

Using the affine structure $w_t(u)=w_{\mathrm{ini},t}+\col(0,0,u)$ from
\eqref{eq:Z-wt} and the fact that the future-input block of $Z$ is $U_f$, we have
\begin{equation}
\lambda\,Z^\top w_t(u)=c_t+\lambda\,U_f^\top u,
\quad
c_t:=\lambda\,Z^\top w_{\mathrm{ini},t},
\label{eq:ct}
\end{equation}
where $c_t\in\mathbb{R}^{n_g}$ is the \emph{past-data forcing vector}: it
depends only on the measured window $w_{\mathrm{ini},t}$ and, at this stage,
appears to be a per-step quantity requiring a matrix--vector product against
the full regressor $Z$. As shown below, $c_t$ and the $n_g$-dimensional solve
it feeds can be removed entirely from the online phase. Substituting
\eqref{eq:ct} into \eqref{eq:gstar},
\begin{equation}
g^\star(u,y)=M^{-1}\big(c_t+\lambda\,U_f^\top u+\rho\,Y_f^\top y\big).
\label{eq:gstar-ct}
\end{equation}

\begin{remark}[Quadratic vs.\ $\ell_1$ regularization]
\label{rem:l1l2}
The regularized DeePC formulation of \cite[Eq.~(8)]{coulson2019data} uses
the $\ell_1$ penalty $\lambda_g\|g\|_1$. Although $\ell_1$-regularized
least squares can itself be reformulated as a QP, this reformulation
retains $g$ (or its split positive/negative parts) as an online decision
variable and admits no closed-form elimination analogous to
\eqref{eq:gstar-ct}; it is the absence of this elimination, not the
non-smoothness of $\|g\|_1$ per se, that prevents removing $g$ from the
online problem. Replacing this term by the quadratic ridge
$\tfrac{\lambda_g}{2}\|g\|_2^2$ instead yields \eqref{eq:gstar-ct}, which
enables the offline elimination of $g$ and the reduced BoxQP. Since
$M\succ0$ for $\lambda,\rho,\lambda_g>0$, $g^\star$ varies smoothly with
these positive weights. Quadratic-regularized DeePC likewise admits a
closed-form solution~\cite{huang2021quadratic}. The $\ell_1$-regularized
hard formulation and the proposed quadratic soft formulation are different
optimization problems, so they can produce different closed-loop
trajectories, as illustrated in Section~\ref{sec:results}.
\end{remark}

\subsection{Substituting $g^\star$ back}

Inserting $g^\star(u,y)$ from \eqref{eq:gstar-ct} into the objective of
\eqref{eq:penalized} and discarding terms independent of $(u,y)$ produces a cost
in the physical variables alone. Writing $z:=\col(u,y)\in\mathbb{R}^{N(m+p)}$,
the reduced cost is the quadratic
\begin{equation}
J_{\mathrm{red}}(z)=\tfrac12\,z^\top\widehat H\,z+\widehat h_t^\top z+\mathrm{const},
\label{eq:Jred}
\end{equation}
with symmetric block Hessian and linear term
\begin{equation}
\widehat H=
\begin{bmatrix}\widehat H_{uu} & \widehat H_{uy}\\[2pt]
              \widehat H_{uy}^\top & \widehat H_{yy}\end{bmatrix},
\quad
\widehat h_t=
\begin{bmatrix}\widehat h_u\\[2pt]\widehat h_y\end{bmatrix},
\label{eq:Hhat-hhat}
\end{equation}
whose blocks follow directly from \eqref{eq:gstar-ct}:
\begin{align}
\widehat H_{uu}&=W_u+\lambda I-\lambda^2\,U_f M^{-1}U_f^\top,
\label{eq:Huu}\\
\widehat H_{yy}&=W_y+\rho I-\rho^2\,Y_f M^{-1}Y_f^\top,
\label{eq:Hyy}\\
\widehat H_{uy}&=-\lambda\rho\,U_f M^{-1}Y_f^\top,
\label{eq:Huy}\\
\widehat h_u&=-\lambda\,U_f M^{-1}c_t,
\label{eq:hu}\\
\widehat h_y&=-W_y r-\rho\,Y_f M^{-1}c_t.
\label{eq:hy}
\end{align}
The terms $+\lambda I$ and $+\rho I$ arise from the squared targets
$\tfrac{\lambda}{2}\|w_t(u)\|^2$ and $\tfrac{\rho}{2}\|y\|^2$, while the negative
terms $-\lambda^2 U_fM^{-1}U_f^\top$, $-\rho^2 Y_fM^{-1}Y_f^\top$, and
$-\lambda\rho\,U_fM^{-1}Y_f^\top$ are the data-driven curvature contributed by
the eliminated coordinate $g$.

\subsubsection{$T$-independent online construction via a Woodbury reformulation}
\label{sec:online-elim}

Although $\widehat H$ is constant and $\widehat h_t$ is updated online, their direct construction involves $M\in\mathbb{R}^{n_g\times n_g}$, whose dimension grows with the offline data length $T$. We therefore derive an algebraically equivalent construction whose factorization dimension is independent of $T$. This changes only the offline construction and leaves the soft quadratic DeePC problem \eqref{eq:penalized} and its optimizer unchanged.

Define $f_t:=w_{\mathrm p,t}=\operatorname{col}(u_{\mathrm{ini}},y_{\mathrm{ini}})$
and $z:=\operatorname{col}(u,y)$, and collect the Hankel blocks and weights as
\begin{equation}
\begin{aligned}
A_f&:=\begin{bmatrix}U_p\\Y_p\end{bmatrix},\quad
A_z:=\begin{bmatrix}U_f\\Y_f\end{bmatrix},\quad
A:=\begin{bmatrix}A_f\\A_z\end{bmatrix},\\
\Sigma&:=\operatorname{blkdiag}\bigl(\lambda I_{T_{\mathrm{ini}}(m+p)},\lambda I_{Nm},\rho I_{Np}\bigr).
\end{aligned}
\label{eq:woodbury-A}
\end{equation}
The $g$-dependent terms of \eqref{eq:penalized} are then
$\tfrac12\|Ag-\operatorname{col}(f_t,z)\|_\Sigma^2+\tfrac{\lambda_g}{2}\|g\|_2^2$. By the Woodbury identity, minimum over $g$ is
$\tfrac12\operatorname{col}(f_t,z)^\top K\operatorname{col}(f_t,z)$ with
\begin{equation}
K:=\Bigl(\Sigma^{-1}+\tfrac{1}{\lambda_g}AA^\top\Bigr)^{-1}
=\begin{bmatrix}K_{ff}&K_{fz}\\K_{zf}&K_{zz}\end{bmatrix}.
\label{eq:woodbury-K}
\end{equation}
With $H_0:=\operatorname{blkdiag}(W_u,W_y)$ and
$h_{0,t}:=\operatorname{col}(0,-W_yr_t)$ from the tracking cost, expanding in
$z$ gives an alternative to \eqref{eq:Huu}--\eqref{eq:hy}:
\begin{equation}
\widehat H=H_0+K_{zz},
\qquad
\widehat h_t=K_{zf}f_t+h_{0,t}.
\label{eq:woodbury-reduced}
\end{equation}

The implementation never forms $K$. Offline it factors
$G:=\Sigma^{-1}+\lambda_g^{-1}AA^\top$ and solves
$G\,\operatorname{col}(X_f,X_z)=\operatorname{col}(0,I)$; since this recovers
the $z$-columns of $K$, symmetry gives $K_{zf}=X_f^\top$, $K_{zz}=X_z$, and
only $\widehat H$ and $K_{zf}$ are stored. At time $t$,
\begin{equation}
\widehat h_t=K_{zf}f_t+\operatorname{col}(0,-W_yr_t),
\label{eq:hu-hy-fast}
\end{equation}
The dimensions of $G$, $\widehat H$, and $K_{zf}$ depend only on $N,T_{\mathrm{ini}},m,p$, not on $T$. Therefore the online phase never solves an $n_g$-dimensional system nor reconstructs $g^\star$. Measured exogenous inputs are handled by appending the corresponding signal and Hankel rows to $f_t$ and $A_f$.

\begin{remark}[Online fixed dimensions]
\label{rem:scalable}
The matrix $G$, the reduced Hessian $\widehat H\in\mathbb{R}^{N(m+p)\times N(m+p)}$, and the map $K_{zf}\in\mathbb{R}^{N(m+p)\times T_{\mathrm{ini}}(m+p)}$ all have dimensions independent of the offline data length $T$. Hence, for fixed $N$, $m$, $p$, and $T_{\mathrm{ini}}$, all online matrix--vector products have $T$-independent dimensions. The $T$-dependent construction is moved entirely to the offline phase.
\end{remark}

Rescaling the box sets $\mathcal U^N,\mathcal Y^N$ to$[-\mathbf 1,\mathbf 1]$ by the affine change of variables of~\cite{wu2025direct}, the relaxed DeePC program\eqref{eq:penalized} therefore reduces to the box-constrained quadratic program
\begin{equation}
\min_{z}\ \tfrac12\,z^\top\widehat H\,z+\widehat h_t^\top z
\quad\text{s.t.}\quad -\mathbf 1\le z\le\mathbf 1 .
\label{eq:boxqp}
\end{equation}

Introducing the upper- and lower-bound multipliers $\gamma,\theta\ge0$ and the corresponding slacks $\phi,\psi\ge0$ (so that $\phi=\mathbf 1-z$ is the upper-bound slack paired with $\gamma$, and $\psi=z+\mathbf 1$ is the lower-bound slack paired with $\theta$), the Karush--Kuhn--Tucker (KKT)
conditions of \eqref{eq:boxqp} are
\begin{equation}
\begin{aligned}
&\widehat H z+\widehat h_t+\gamma-\theta=0,~z+\phi-\mathbf 1=0,~z-\psi+\mathbf 1=0,\\
&\gamma,\theta,\phi,\psi\ge0,
~\gamma\odot\phi=0,
~\theta\odot\psi=0,
\end{aligned}
\label{eq:kkt}
\end{equation}
where $\odot$ denotes the elementwise product; \eqref{eq:kkt} is the standard Box-QP optimality system of~\cite{wu2025direct}. We solve
\eqref{eq:boxqp} with the feasible Mehrotra predictor--corrector IPM of Algorithm~\ref{alg:boxqp-ipm}: a predictor--corrector IPM for the same problem class with a certified $O(\sqrt n)$ iteration bound has recently been proposed~\cite{wu2026boxQP}.

\begin{remark}[Box constraints with costs]
\label{rem:scope}
Unlike the original DeePC~\cite{coulson2019data}, whose constraint sets $\mathcal{U},\mathcal{Y}$ may be arbitrary convex sets when solved by a general-purpose solver, the reduced program obtained after eliminating $g$ remains a convex problem for any convex $\mathcal U,\mathcal Y$, but remains a \emph{quadratic program} only when these sets are polyhedral. The specialized online BoxQP solver considered here further requires componentwise box constraints; coupled polyhedral constraints such as $Az\le b$ or nonlinear constraints such as $y_1^2+y_2^2\le1$ require a more general QP or nonlinear-programming solver. This concerns \textbf{only the constraint set} of the online problem, not the system itself, since the plant enters only through the recorded data and the soft consistency penalties~\eqref{eq:penalized}.
\end{remark}

\begin{remark}[Dense Newton system]
\label{rem:dense-hessian}
The proposed BoxQP has the same bound constraints and Mehrotra predictor--corrector IPM structure as \cite{wu2026koopman}, but a different Hessian structure. In particular, $\widehat H$ is generally dense even when $H_0$ is diagonal; therefore, the diagonal-block Schur elimination of \cite{wu2026koopman} is not applicable. We retain the same KKT equations, but factor the full dense Newton matrix
\[
\widehat H+D(\gamma/\phi+\theta/\psi)
\]
by Cholesky at each IPM iteration and reuse the factor for the predictor and corrector directions.

Since $\widehat H$ is constant and only the barrier diagonal varies, the Newton matrix has a fixed-plus-diagonal form. This structure could instead use a bounded sequence of rank-one inverse updates as in \cite{wu2026n3QP}, reducing the per-iteration linear-algebra cost from $O(n^3)$ to $O(n^2)$. We retain dense Cholesky in this work and leave this substitution for future work.
\end{remark}

\subsubsection{Strong convexity and penalty consistency}
\label{sec:feas}

\begin{theorem}[Strong convexity and spectral bounds]
\label{thm:feas}
Let $\lambda,\rho,\lambda_g>0$, $W_u,W_y\succeq0$, and assume
$\mathcal U^N\times\mathcal Y^N$ is a nonempty closed box. Recall from
\eqref{eq:woodbury-A}--\eqref{eq:woodbury-reduced} that
$\Sigma=\operatorname{blkdiag}(\lambda I_{T_{\mathrm{ini}}(m+p)},\lambda I_{Nm},\rho I_{Np})$,
$H_0=\operatorname{blkdiag}(W_u,W_y)\succeq0$, and that $K_{zz}$ is the
principal submatrix of $K=(\Sigma^{-1}+\lambda_g^{-1}AA^\top)^{-1}$ associated with $z$. Then the reduced Hessian $\widehat H=H_0+K_{zz}$ satisfies
\begin{equation}
\begin{aligned}
\lambda_{\min}(\widehat H)&\ge\lambda_{\min}(H_0)
+\Bigl(\tfrac{1}{\min(\lambda,\rho)}+\tfrac{\lVert A\rVert_2^2}{\lambda_g}\Bigr)^{-1}>0,\\
\lambda_{\max}(\widehat H)&\le\lambda_{\max}(H_0)+\max(\lambda,\rho),
\end{aligned}
\label{eq:spectral-bounds}
\end{equation}
where $\lVert A\rVert_2$ is the spectral norm of $A$. Consequently, the reduced QP \eqref{eq:boxqp} admits a unique minimizer for any recorded data.
\end{theorem}

\begin{proof}
Let $G=\Sigma^{-1}+\lambda_g^{-1}AA^\top$. Since $G\succeq\Sigma^{-1}\succ0$,
$K=G^{-1}\preceq\Sigma$, so $\lambda_{\max}(K)\le\max(\lambda,\rho)$; and
$\lambda_{\max}(G)\le1/\min(\lambda,\rho)+\lVert A\rVert_2^2/\lambda_g$
bounds $\lambda_{\min}(K)=1/\lambda_{\max}(G)$ from below. Because $K_{zz}$
is a principal submatrix of $K$, eigenvalue interlacing gives
$\lambda_{\min}(K)\le\lambda_{\min}(K_{zz})$ and
$\lambda_{\max}(K_{zz})\le\lambda_{\max}(K)$, and Weyl's inequalities
applied to $\widehat H=H_0+K_{zz}$ yield \eqref{eq:spectral-bounds}. The
objective is therefore strongly convex and coercive, which result in a unique minimizer over the nonempty closed convex feasible set.
\end{proof}

\begin{proposition}[Lipschitz sensitivity of the optimizer]
\label{prop:lipschitz}
Let $\mu:=\lambda_{\min}(\widehat H)$, and let $z_i^\star$ denote the
minimizer of \eqref{eq:boxqp} for measured window $f_i$ and reference
$r_i$, $i=1,2$, where $f_t=w_{\mathrm p,t}$ is the measured window itself.
Then
\begin{equation}
\lVert z_1^\star-z_2^\star\rVert_2\le\frac{1}{\mu}
\bigl(\lVert K_{zf}\rVert_2\lVert f_1-f_2\rVert_2
+\lVert W_y\rVert_2\lVert r_1-r_2\rVert_2\bigr),
\label{eq:lipschitz}
\end{equation}
with $\lVert K_{zf}\rVert_2\le\max(\lambda,\rho)$.
The applied input is the first block of $z^\star$ and therefore satisfies the same bound. Hence, small changes in the measured window or reference produce proportionally small changes in the control action.
\end{proposition}

\begin{proof}
Optimality over the closed convex box gives
$(\widehat Hz_i^\star+\widehat h_i)^\top(z-z_i^\star)\ge0$ for all feasible
$z$. Taking $z=z_j^\star$, $j\ne i$, and adding the two inequalities yields
$\mu\lVert z_1^\star-z_2^\star\rVert_2^2
\le\lVert\widehat h_1-\widehat h_2\rVert_2\lVert z_1^\star-z_2^\star\rVert_2$.
Substituting $\widehat h_i=K_{zf}f_i+\operatorname{col}(0,-W_yr_i)$ gives
\eqref{eq:lipschitz}, and $\lVert K_{zf}\rVert_2\le\lVert K\rVert_2\le\max(\lambda,\rho)$.
\end{proof}

\begin{proposition}[Residual bound]
\label{cor:residual}
Suppose there exists $(\bar g,\bar u,\bar y)$ with $Z\bar g=w_t(\bar u)$, $Y_f\bar g=\bar y$, $\bar u\in\mathcal U^N$, $\bar y\in\mathcal Y^N$, and let $J_{\mathrm{feas}}$ be the value of the full objective of \eqref{eq:penalized} at this point, including all constant terms rather than the reduced cost \eqref{eq:Jred}; since its penalty terms vanish,
$J_{\mathrm{feas}}$ is independent of $\lambda,\rho$. Then any minimizer
$(g^\star,u^\star,y^\star)$ of \eqref{eq:penalized} satisfies
\begin{equation}
\begin{aligned}
\lVert Zg^\star-w_t(u^\star)\rVert_2\le\sqrt{\frac{2J_{\mathrm{feas}}}{\lambda}},
\qquad\\
\operatorname{dist}\bigl(Y_fg^\star,\mathcal Y^N\bigr)
\le\lVert Y_fg^\star-y^\star\rVert_2\le\sqrt{\frac{2J_{\mathrm{feas}}}{\rho}}.
\label{eq:residual-bound}
\end{aligned}
\end{equation}
The second bound shows that the behavioral prediction is restricted within
$\sqrt{2J_{\mathrm{feas}}/\rho}$ of $\mathcal Y^N$.
\end{proposition}
\begin{proof}
Every term of the full objective of \eqref{eq:penalized} is a nonnegative
quadratic, and its optimal value is at most $J_{\mathrm{feas}}$, so
$\tfrac{\lambda}{2}\lVert Zg^\star-w_t(u^\star)\rVert_2^2\le J_{\mathrm{feas}}$
and $\tfrac{\rho}{2}\lVert Y_fg^\star-y^\star\rVert_2^2\le J_{\mathrm{feas}}$.
The distance bound holds because $y^\star\in\mathcal Y^N$.
\end{proof}

\begin{proposition}[Soft-to-hard limit]
\label{prop:limit}
Under the assumption of Proposition~\ref{cor:residual}, with
$\mathcal U^N,\mathcal Y^N$ compact and $\lambda_g$ fixed, every limit point of $(g^\star,u^\star,y^\star)$ along any sequence $(\lambda_k,\rho_k)$ satisfying $\min\{\lambda_k,\rho_k\}\to\infty$ satisfies \eqref{eq:consistency} exactly and
solves the hard quadratically regularized counterpart of \eqref{eq:deepc-qp}, obtained by adding the ridge $\tfrac{\lambda_g}{2}\lVert g\rVert_2^2$ to its objective.
\end{proposition}

\begin{proof}
By Proposition~\ref{cor:residual}, both residuals vanish along the sequence. Since $\tfrac{\lambda_g}{2}\lVert g^\star\rVert_2^2\le J_{\mathrm{feas}}$ and $(u^\star,y^\star)$ lies in the bounded closed set
$\mathcal U^N\times\mathcal Y^N$, the minimizers are bounded; let $(g^\dagger,u^\dagger,y^\dagger)$ be the limit of a convergent subsequence. It satisfies \eqref{eq:consistency} because the residuals vanish. For any hard-feasible $(g,u,y)$, optimality and nonnegativity of the penalties
give $J(u^\star,y^\star)+\tfrac{\lambda_g}{2}\lVert g^\star\rVert_2^2
\le J(u,y)+\tfrac{\lambda_g}{2}\lVert g\rVert_2^2$, where $J$ is the tracking-and-effort cost; passing to the limit shows that
$(g^\dagger,u^\dagger,y^\dagger)$ is optimal for the hard problem.
\end{proof}

\section{Numerical Examples}
\label{sec:results}

We evaluate the proposed Fast-DeePC scheme on three benchmarks. Each benchmark compares the original hard-equality $\ell_1$ controller as a control-performance reference, with different solve structure for the same soft quadratic problem: Problem \eqref{eq:penalized} solve with full CVXPY/OSQP, the reduced BoxQP \eqref{eq:boxqp}solved by OSQP, and by Algorithm~\ref{alg:boxqp-ipm}. Comparisons among these simulations evaluate both the formulation and solver acceleration. Table~\ref{tab:settings} lists the benchmark dimensions, horizons, and penalty weights used throughout. Simulations are implemented in Python on a desktop with an Intel i7-11800H CPU and an NVIDIA GeForce RTX 3060 GPU. Codes are available at https://github.com/TerryZhajimi/Fast-DeePC.

\begingroup
\setlength{\abovedisplayskip}{2pt}
\setlength{\belowdisplayskip}{2pt}

\begin{algorithm}[!t]
\footnotesize
\caption{Mehrotra predictor--corrector IPM for the reduced BoxQP}
\label{alg:boxqp-ipm}
\KwIn{$\widehat H$, $\widehat h_t$, tolerance $\epsilon$,
maximum iterations $N_{\max}$, optional warm start $z_{\mathrm{warm}}$.}
\KwOut{Approximate minimizer of
$\frac12z^\top\widehat Hz+\widehat h_t^\top z$ subject to
$-\mathbf1\le z\le\mathbf1$.}

$n\leftarrow\dim(z)$; initialize $z$ by clipping $z_{\mathrm{warm}}$
strictly inside the box, or set $z\leftarrow0$\;
$\phi\leftarrow\mathbf1-z$, \quad $\psi\leftarrow\mathbf1+z$\;

$q\leftarrow\widehat Hz+\widehat h_t$,\quad
$\eta\leftarrow\max\{1,\lVert q\rVert_\infty\}$\;
$\gamma\leftarrow\eta\mathbf1-\frac12q$,\quad
$\theta\leftarrow\eta\mathbf1+\frac12q$\;

\For{$k=0,\ldots,N_{\max}-1$}{
    $\mu\leftarrow
    (\gamma^\top\phi+\theta^\top\psi)/(2n)$\;
    \lIf{$\mu\le\epsilon$}{\Return{$z$}}

    $\bar H\leftarrow\widehat H+
    \operatorname{diag}(\gamma\oslash\phi+\theta\oslash\psi)$;
    factor $\bar H=LL^\top$\;

    \tcp{Affine predictor}
    $r_1^{\mathrm{aff}}\leftarrow-\gamma\odot\phi$,\quad
    $r_2^{\mathrm{aff}}\leftarrow-\theta\odot\psi$\;

    Solve
    $\bar H\Delta z^{\mathrm{aff}}
    =r_2^{\mathrm{aff}}\oslash\psi
    -r_1^{\mathrm{aff}}\oslash\phi$\;

    $\Delta\phi^{\mathrm{aff}}\leftarrow-\Delta z^{\mathrm{aff}}$,\quad
    $\Delta\psi^{\mathrm{aff}}\leftarrow\Delta z^{\mathrm{aff}}$\;
    $\Delta\gamma^{\mathrm{aff}}\leftarrow
    (\gamma\oslash\phi)\odot\Delta z^{\mathrm{aff}}
    +r_1^{\mathrm{aff}}\oslash\phi$\;
    $\Delta\theta^{\mathrm{aff}}\leftarrow
    -(\theta\oslash\psi)\odot\Delta z^{\mathrm{aff}}
    +r_2^{\mathrm{aff}}\oslash\psi$\;

    $\alpha^{\mathrm{aff}}\leftarrow
    \min\!\left\{1,\,
    0.99\min_{\substack{v\in\{\gamma,\theta,\phi,\psi\}\\
                        \Delta v_i^{\mathrm{aff}}<0}}
    \left(-v_i/\Delta v_i^{\mathrm{aff}}\right)\right\}$\;

    $v_+^{\mathrm{aff}}\leftarrow
    v+\alpha^{\mathrm{aff}}\Delta v^{\mathrm{aff}}$
    for $v\in\{\gamma,\theta,\phi,\psi\}$\;
    $\mu^{\mathrm{aff}}\leftarrow
    [(\gamma_+^{\mathrm{aff}})^\top\phi_+^{\mathrm{aff}}
    +(\theta_+^{\mathrm{aff}})^\top\psi_+^{\mathrm{aff}}]/(2n)$\;
    $\sigma\leftarrow\min\{1,(\mu^{\mathrm{aff}}/\mu)^3\}$\;

    \tcp{Centering corrector}
    $r_1\leftarrow\sigma\mu\mathbf1-\gamma\odot\phi
    -\Delta\gamma^{\mathrm{aff}}\odot\Delta\phi^{\mathrm{aff}}$\;
    $r_2\leftarrow\sigma\mu\mathbf1-\theta\odot\psi
    -\Delta\theta^{\mathrm{aff}}\odot\Delta\psi^{\mathrm{aff}}$\;

    Solve $\bar H\Delta z=r_2\oslash\psi-r_1\oslash\phi$ using $L$\;
    $\Delta\phi\leftarrow-\Delta z$,\quad
    $\Delta\psi\leftarrow\Delta z$\;
    $\Delta\gamma\leftarrow
    (\gamma\oslash\phi)\odot\Delta z+r_1\oslash\phi$\;
    $\Delta\theta\leftarrow
    -(\theta\oslash\psi)\odot\Delta z+r_2\oslash\psi$\;

    $\alpha\leftarrow
    \min\!\left\{1,\,
    0.99\min_{\substack{v\in\{\gamma,\theta,\phi,\psi\}\\
                        \Delta v_i<0}}
    \left(-v_i/\Delta v_i\right)\right\}$\;

    $(z,\gamma,\theta,\phi,\psi)\leftarrow
    (z,\gamma,\theta,\phi,\psi)
    +\alpha(\Delta z,\Delta\gamma,\Delta\theta,\Delta\phi,\Delta\psi)$\;
}
\Return{$z$}\;
\end{algorithm}
\endgroup

\begin{table}[t]
\centering
\footnotesize
\setlength{\tabcolsep}{4pt}
\caption{Benchmark settings.}
\label{tab:settings}
\begin{tabular}{@{}lrrrrrrrr@{}}
\toprule
Benchmark & $T$ & $T_{\mathrm{ini}}$ & $N$ & $m$ & $p$ & $\lambda{=}\rho$ & $\lambda_g$ & $N(m{+}p)$ \\
\midrule
Rocket    & 225  & 1  & 10 & 3 & 6  & $10^5$ & 281.17 & 90 \\
Two-tank  & 1500 & 2  & 50 & 2 & 2  & $10^6$ & 2      & 200 \\
Traffic   & 1200 & 20 & 50 & 2 & 20 & $10^4$ & 20     & 600 \\
\bottomrule
\end{tabular}
\end{table}

\subsection{Nonlinear rocket landing}
The rocket benchmark uses the DeePC-Hunt VTVL landing simulator~\cite{cummins2024deepc}, with state $x=\operatorname{col}(p_x,p_y,\dot p_x,\dot p_y,\theta,\dot\theta)\in\mathbb{R}^6$ (planar position, velocity, pitch, and pitch rate) and input $u=\operatorname{col}(F_E,F_S,\phi)\in\mathbb{R}^3$ (normalized main thrust, side thrust, and gimbal angle), with full-state output $p=6$. With full-state feedback and $T_{\mathrm{ini}}=1$, this benchmark is effectively a state-feedback predictive-control problem rather than a
test of state estimation from output data alone. Tracking weights are $W_y=2(I_N\otimes Q)$, $W_u=2(I_N\otimes R)$ with $Q=\operatorname{diag}(100,15,25,2.2,6000,120)$ and $R=\operatorname{diag}(0.01,0.30,0.30)$; the $\lambda_g$ are selected by the differentiable-optimization training procedure of DeePC-Hunt. The predicted pitch is constrained to $|\theta|\le0.45$ rad, a safety margin relative to the physical tipping limit $|\theta|\le0.6$ rad. The original DeePC-Hunt baseline retains its original DeePC-Hunt formulation with its trained regularizers.

For the original DeePC-Hunt baseline, we retain its distinct formulation, original tracking weights, and trained regularizers as a control-performance baseline. The full quadratic CVXPY/OSQP and Fast-DeePC controllers solve the identical soft-constrained problem.

All three controllers are simulated from the same initial state and land successfully. Figure~\ref{fig:rocket-traj} shows the descent paths of the original DeePC-Hunt controller and the proposed Fast-DeePC Scheme. The proposed solver and the CVXPY baseline solve the identical soft-constrained problem and achieve a closed-loop trajectory RMSE of $7\times10^{-5}$, confirming that the offline elimination of $g$ reproduces the full penalized solution up to solver tolerance.

\begin{figure}[t]
\centering
\includegraphics[width=\linewidth]{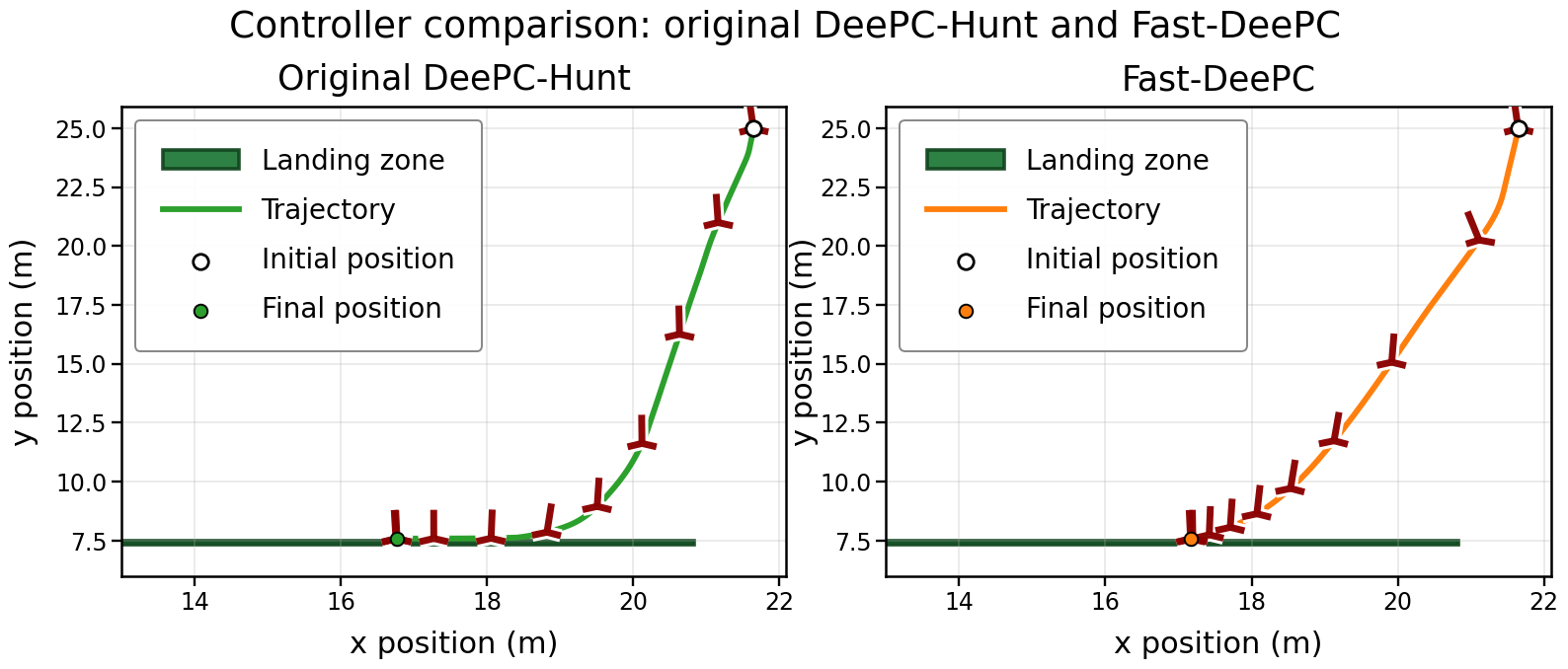}
\caption{Rocket closed-loop all states landing trajectory simulation}
\label{fig:rocket-traj}
\end{figure}

\subsection{Two-tank system}
The two-tank plant is linear, with fixed matrices $A,B,C,D$; its $n=2$ states are the two tank levels, with $m=2$ flow-command inputs and $p=2$ measured outputs. Tracking weights are $W_y=2I_{100}$, $W_u=0$ (both channels weighted equally, no input-effort term); both controllers enforce $u\in[-1,1]$, $y\in[-1,1]$, and the $\ell_1$ baseline uses the same $\lambda_g=2$.

Figure~\ref{fig:twotank-control} shows that Fast-DeePC tracks the reference more closely than the original hard-$\ell_1$ DeePC baseline in this experiment. Fast-DeePC reduces the tracking RMSE from $0.209$ to $0.184$ and the common closed-loop cost from $2.63$ to $2.04$, while the realized outputs of both controllers remain within their prescribed bounds. Separately, the reduced native-OSQP and C solutions have a closed-loop trajectory RMSE of $0.007$, providing a numerical check of the C implementation. The mean behavioral residuals are below $10^{-6}$, confirming that the soft solution remains close to the underlying behavioral relations.

\begin{figure}[t]
\centering
\includegraphics[width=\linewidth]{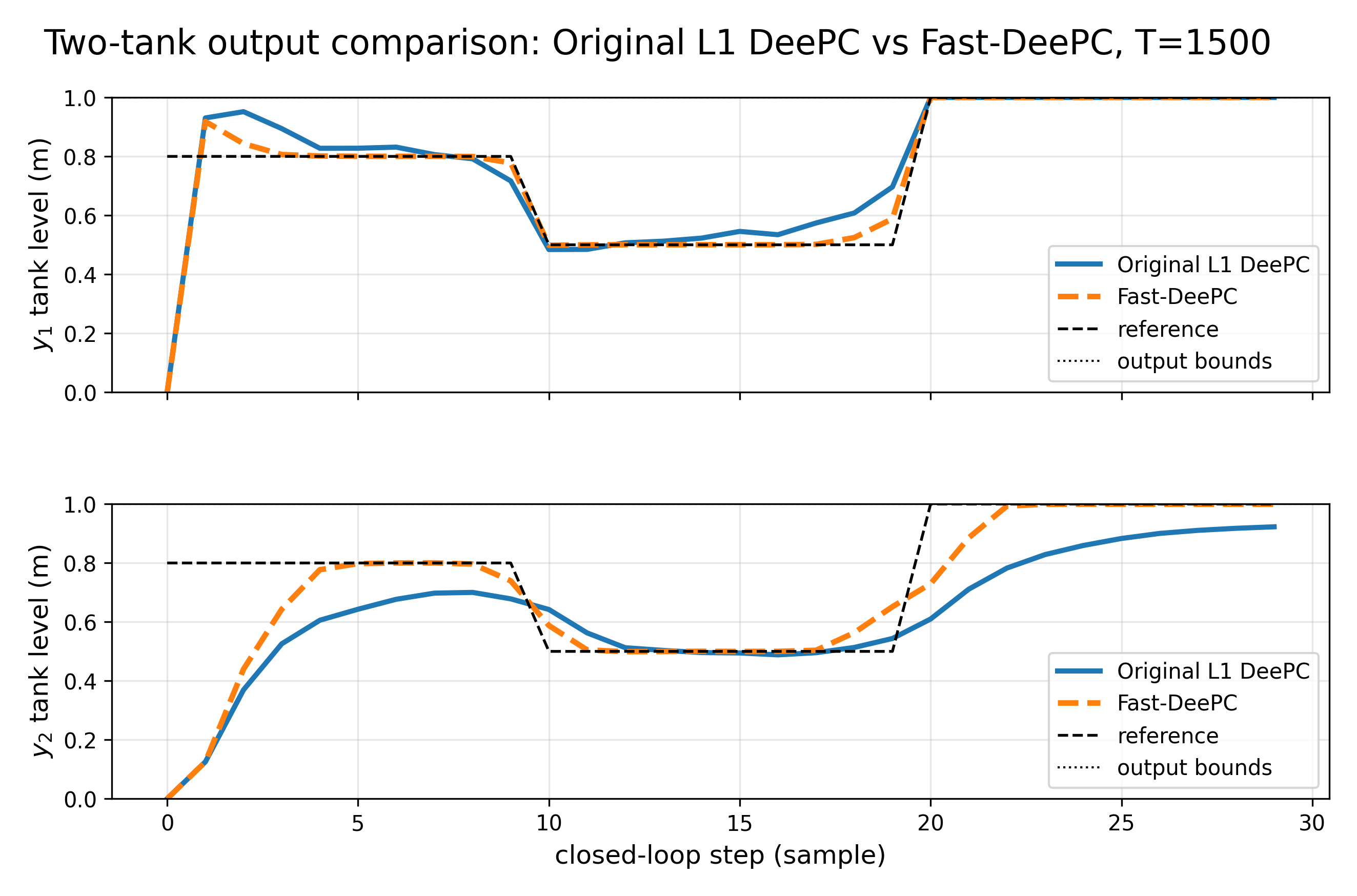}
\vspace{-20pt}
\caption{Two-tank output comparison tracking against the reference}
\label{fig:twotank-control}
\end{figure}

\subsection{Mixed-traffic platoon}
The eight-vehicle mixed-traffic DeeP-LCC platoon~\cite{wang2023deep} has two connected autonomous vehicles (CAVs), at positions 3 and 6, and six human-driven vehicles (HDVs) following nonlinear optimal-velocity-model dynamics. The control input is the CAV acceleration $u\in\mathbb{R}^2$; the head-vehicle velocity is a measured exogenous signal, appended as an additional Hankel block alongside $U_p,Y_p,U_f,Y_f$ following the original DeeP-LCC construction. The output ($p=10$) is the velocity error of all eight vehicles together with the spacing error of the two CAVs relative to the equilibrium spacing $s^\star=20$ m, so the absolute spacing bound $[5,40]$ m corresponds to a spacing-error bound of $[-15,20]$ m; sampling period is $\Delta t=0.05$ s. Stage-cost weights are $w_v=1$, $w_s=0.5$, $w_u=0.1$, matching the original DeeP-LCC weighting; CAV accelerations are constrained to $[-5,2]~\mathrm{m/s^2}$.

Figure~\ref{fig:lcc-control} shows comparable tracking RMSE and cost between Fast-DeePC and the hard-$\ell_1$ baseline. The hard-$\ell_1$ baseline required $28$ automatic-emergency-braking interventions (an external safety layer triggered whenever its optimized inputs alone were about to violate the spacing constraint)  to avoid collision. Fast-DeePC, with automatic-emergency-braking disabled throughout, avoided collision using its optimized inputs alone.

\begin{figure}[t]
\centering
\includegraphics[width=\linewidth]{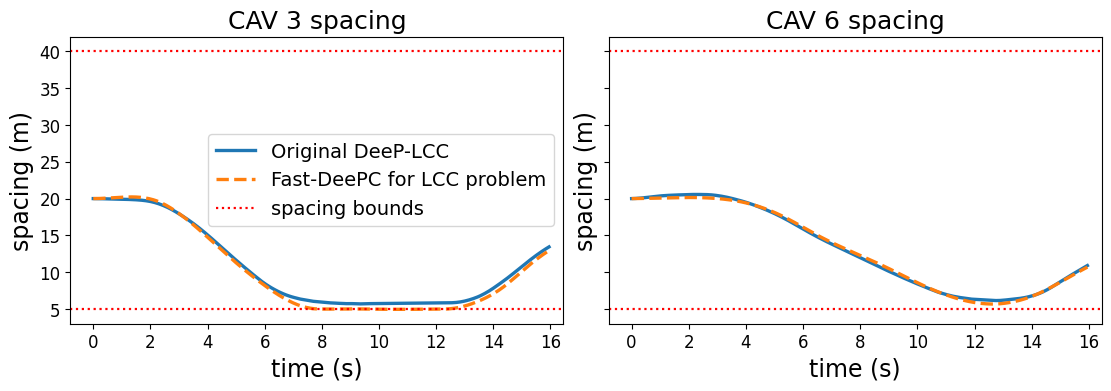}
\vspace{-15pt}
\caption{CAV spacing under a traffic disturbance with all controllers}
\label{fig:lcc-control}
\end{figure}

\begin{figure}[t]
\centering
\includegraphics[width=\linewidth]{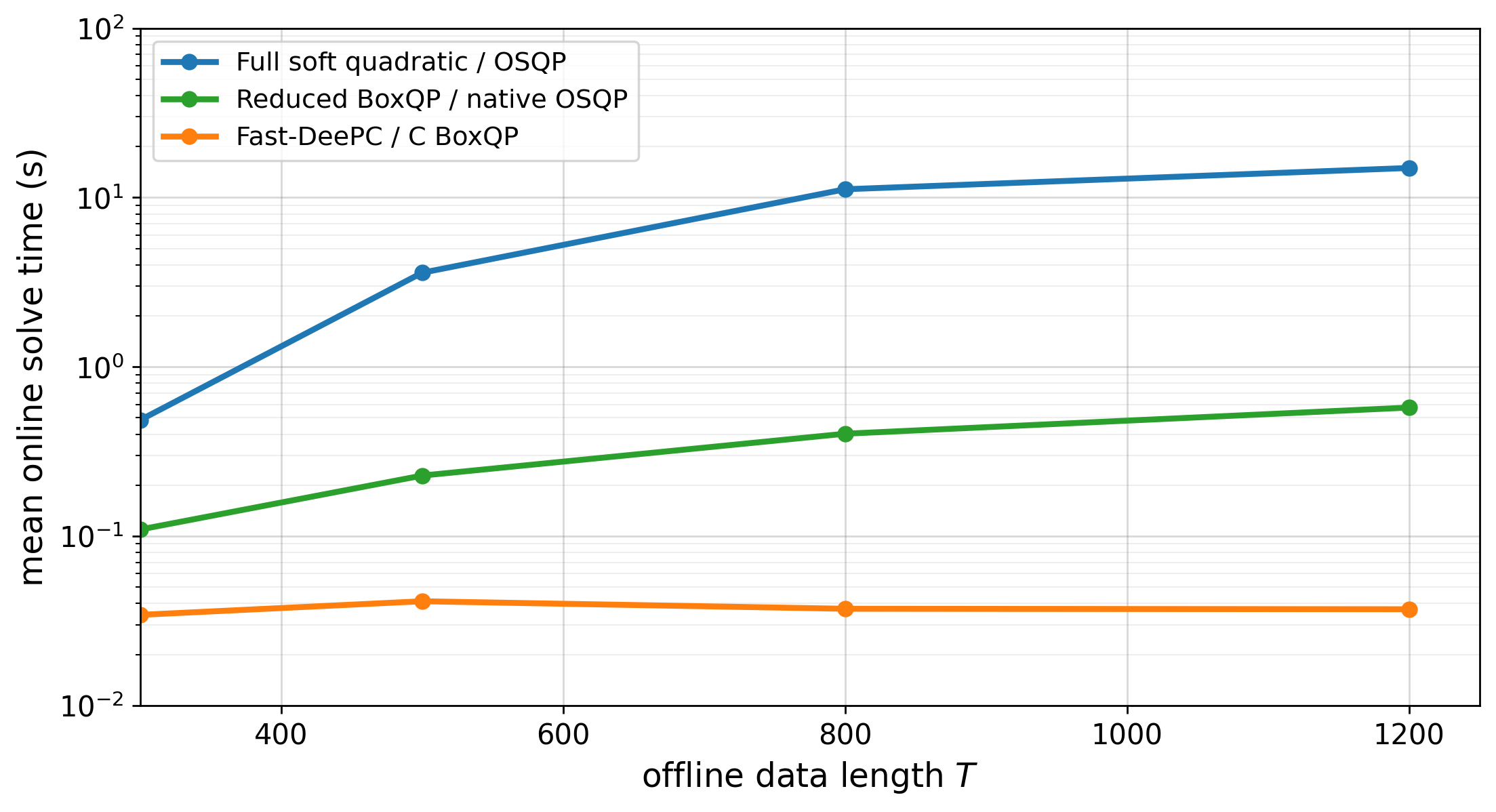}
\vspace{-15pt}
\caption{Mean per-step online solve time with increasing offline data
length $T$ on the mixed-traffic DeeP-LCC benchmark.}
\label{fig:scaling-traffic}
\end{figure}

Figure~\ref{fig:scaling-traffic} and Table~\ref{tab:speedup} separate the effects of eliminating $g$ and using the specialized solver. Relative to full soft OSQP, the C BoxQP solver is approximately $19\times$ faster on the two-tank benchmark and more than two orders of magnitude faster on the mixed-traffic benchmark. On the smaller rocket problem, the C solver is approximately $8.5\times$ faster than full soft OSQP and exhibits fewer iterations and lower timing variability.

\begin{table}[!htbp]
\centering
\footnotesize
\setlength{\tabcolsep}{3pt}
\renewcommand{\arraystretch}{1.08}
\caption{Online solve time (s) and solver iterations.}
\label{tab:speedup}
\begin{tabular}{@{}llrrr@{}}
\toprule
Benchmark & Method & Mean & Std. & Iter.\ (mean/max) \\
\midrule

\multirow{4}{*}{Two-tank}
& Original hard-$\ell_1$ DeePC
& 1.3506 & 0.0666 & 13.43/15 \\
& Full soft-$\ell_2$/OSQP
& 0.2303 & 0.1253 & 317.50/950 \\
& Reduced/OSQP
& 0.2372 & 0.1850 & 6000/15000 \\
& Reduced/Fast-DeePC
& \textbf{0.0121} & \textbf{0.0034} & 11.63/15 \\
\midrule

\multirow{4}{*}{Rocket}
& Original DeePC-Hunt
& 0.0227 & 0.0126 & 76/425 \\
& Full soft-$\ell_2$/OSQP
& 0.0058 & 0.0010 & 49.50/50 \\
& Reduced/OSQP
& 0.00058 & 0.00045 & 51.58/300 \\
& Reduced/Fast-DeePC
& \textbf{0.00068} & \textbf{0.00017} & 5.66/11 \\
\midrule

\multirow{4}{*}{Traffic}
& Original hard-$\ell_1$ DeeP-LCC
& 22.0861 & 2.8949 & 11.90/14 \\
& Full soft-$\ell_2$/OSQP
& 4.7027 & 7.0354 & 1040.78/8575 \\
& Reduced/OSQP
& 0.2203 & 0.1672 & 445/2400 \\
& Reduced/Fast-DeePC
& \textbf{0.0330} & \textbf{0.0145} & 9.01/16 \\
\bottomrule
\end{tabular}
\end{table}

\section{Conclusion}
\label{sec:conclusion}

This article presented Fast-DeePC, a soft-constrained reformulation of data-enabled predictive control. By relaxing the behavioral constraint into the objective with a quadratic ridge, Fast-DeePC eliminates the coordinate vector exactly by partial minimization, reducing the online problem to a fixed-dimensional box-constrained QP. The reduced problem admits a unique minimizer for any recorded data, its consistency residuals decay with the penalty weights, and its solution approaches hard quadratically regularized DeePC in the limit. On rocket-landing, two-tank, and mixed-traffic benchmarks, the proposed interior-point solver closely tracks a general-purpose solver on the identical soft problem while reducing online solve time by up to two orders of magnitude. Future work will address streaming data updates, in which the offline matrices are refreshed incrementally as new trajectory data become available, and execution-time certification of the online solver.

\bibliography{Ref}
\end{document}